\documentclass[12pt]{iopart}
\usepackage{iopams} 

\usepackage[english]{babel}
\usepackage{wrapfig} 
\usepackage{tikz}
\usepackage{ifsym}
\usepackage{subcaption}

\usepackage{hyperref}

\usepackage{url}

\expandafter\let\csname equation*\endcsname\relax

\expandafter\let\csname endequation*\endcsname\relax

\usepackage{cite}
\usepackage{amsmath}
\usepackage{amsthm} 
\usepackage{amssymb} 
\usepackage{mathrsfs}
\usepackage{mathtools}
\usepackage{dsfont} 
\usepackage{textcomp} 
\usepackage{physics} 
\newtheorem{theorem}{Theorem}[section]

\newtheorem{corollary}[theorem]{Corollary}
\newtheorem{definition}[theorem]{Definition}

\newtheorem{lemma}[theorem]{Lemma}
\newtheorem{proposition}[theorem]{Proposition}

\theoremstyle{theorem}
\newtheorem{remark}{Remark}

\definecolor{greenlink}{rgb}{0.3, 0.6, 0.3}
\definecolor{greenclear}{rgb}{0.7,0.95,0.7}
\definecolor{grey}{rgb}{0.8,0.8,0.8}
\definecolor{greenother}{rgb}{0.75,0.9,0.75}

\DeclareMathOperator{\Ker}{Ker}
\DeclareMathOperator{\Attr}{Attr}
\DeclareMathOperator\supp{supp}
\DeclareMathOperator\Fix{Fix}
\DeclareMathOperator{\Ran}{Ran}
\DeclareMathOperator{\spec}{spect}

\begin{document}
	\title{\textbf{Asymptotics of quantum channels}}
	\author{Daniele Amato, Paolo Facchi, and Arturo Konderak }
	\address{Dipartimento di Fisica, Universit\`a di Bari, I-70126 Bari, Italy}
	\address{INFN, Sezione di Bari, I-70126 Bari, Italy}

	\begin{abstract}
		We discuss several aspects concerning the asymptotic dynamics of dicrete-time semigroups associated with a quantum channel. By using an explicit expression of the asymptotic map, which describes the action of the quantum channel on its attractor manifold, we investigate the role of permutations in the asymptotic dynamics. We show that, in general, they make the asymptotic evolution non-unitary, and they are related to the divisibility of the quantum channel. Also, we derive several results about the asymptotics of faithful and non-faithful channels, and we establish a constructive unfolding theorem for the asymptotic dynamics.  
	\end{abstract}
	
	\begin{section}{Introduction}
		Modeling a quantum computer or, more generally, any complex quantum hardware by a sequence of unitary gates is no longer sufficient to obtain efficient devices. Indeed, the unavoidable effects of noise and decoherence due to the coupling of the system with its environment is known to be a key obstacle for quantum information processing~\cite{zurek2003decoherence}. 
		
		In this respect, dissipation processes can be harnessed in order to protect quantum information in the manifold of steady states of dissipative systems. In particular, the asymptotic dynamics of open quantum systems can be exploited for the realization of unitary operations inside subspaces protected from noise and decoherence~\cite{zanardi1997noiseless,zanardi2014coherent}. 
		Also, the system evolution at large times may be used in reservoir engineering~\cite{Zoller_res_eng, Wolf_res_eng}, namely the preparation of a target state by relaxation of a system suitably coupled to the environment, or in phase-locking and synchronization of two quantum systems~\cite{jex_synchr}. 	

A non-degenerate attractor manifold, i.e.\ consisting of a unique stationary state towards which the evolution converges, is often not enough in all these applications. 
		
		Therefore, it is crucial to study the dynamics of dissipative systems at large times under less restrictive assumptions, as it was done in~\cite{jex_st_2012,jex_st_2018,wolf2010inverse,wolf2012quantum,baumgartner2012structures,ticozzi_inv,
			albert2019asymptotics} for finite-dimensional systems, and~\cite{carbone2020period} for infinite-dimensional ones. 
		In particular, many efforts were devoted to the asymptotics of Markovian continuous dynamics~\cite{baumgartner2008analysis1,baumgartner2008analysis2,albert2014symmetries,fagnola_2001,agredo2014decoherence}, described by quantum dynamical semigroups. They were completely characterized by Gorini, Kossakowski, Lindblad and Sudarshan (GKLS)~\cite{GKS_76,Lindblad_76}, and the manifold of their steady states was analyzed  in the following years~\cite{Spohn_77,Frigerio_78,Frigerio_Verri_82,Spohn_rev}.  
		
	In this Article we will investigate the asymptotic structure of a quantum channel. The starting point of our analysis is the direct sum decomposition of the attractor manifold, where the asymptotic dynamics takes place, and the action of a quantum channel onto it, as obtained by Wolf and Perez-Garcia~\cite[Theorem 2.1]{wolf2010inverse,wolf2012quantum}.
		In spite of the importance of this result and its use in previous works~\cite{lami2016entanglement,hanson2020eventually}, its consequences have not been fully investigated yet in the literature. 
		
		In particular, the presence of permutations between the factors of the decomposition generally implies the lack of unitarity of the asymptotic dynamics. For a better understanding of the role of permutations, we will find a sufficient condition and a characterization of their absence, from which the connection with the Markovianity and divisibility of the channel emerges.
		
	Also, we will construct a quantum channel with a given asymptotic dynamics, namely with a fixed attractor subspace and dynamics onto it (Theorem 4.1), proving a converse of Theorem 2.1. This is an example of extension theorem for quantum channels, in the trail of similar results  obtained in the literature~\cite{arveson1969subalgebras,jenvcova2012generalized}.
		
		Moreover, we will prove a structure theorem of the asymptotic map in terms of disjoint cyclic permutations (Theorem~\ref{thm:cycles}). It yields an alternative decomposition of the asymptotic manifold which may be obtained from~\cite{guan2018structure}. Finally, previous results on the asymptotics of quantum dynamics by Novotn{\'y} \textit{et al.}~\cite{jex_st_2012,jex_st_2018} will be extended to  non-faithful channels.
		
	The paper is organized as follows. First, we will  recall some general properties of quantum channels (Sec.~\ref{QC}) and of infinitely divisible channels (Sec.~\ref{inf_ch}), and then we give the structure theorem of the asymptotic map in Sec.~\ref{struc_wolf}.
		Afterwards, in Sec.~\ref{aternativeexpression} we prove an alternative structure theorem in terms of disjoint cyclic permutations, while in Section~\ref{extens_th} we will construct a quantum channel unfolding the asymptotic map. Then, in Section~\ref{per_role} we will explore in detail the role of permutations in the the asymptotic dynamics, characterize their absence, 
		and link them to the infinite divisibility and primitivity of the channel. 
		Finally, in Section~\ref{Jex_Wolf} we recall several results on the asymptotics of quantum channels, generalize them to the non-faithful case, and check their consistency with the structure theorems. After the concluding remarks (Section~\ref{concl}), we prove in~\ref{proofs} some results needed in the main part of the paper.         
		
\end{section}
	\begin{section}{Asymptotics of open quantum systems}
		\label{asympt}
		In this Section we will set up the notation and recall some known results about the asymptotic dynamics of finite-dimensional open quantum systems.
		\begin{subsection}{Quantum channels}
			\label{QC}
			The state of an open quantum system is given by a density operator $\rho$, i.e.\ a positive semidefinite operator of unit trace on $\mathcal{H}$, with $\mathcal{H}$ the system Hilbert space, that in the following will be assumed to have a finite dimension $d$. Let $\mathcal{S}(\mathcal{H})$ be the set of density operators on $\mathcal{H}$.
			The evolution of the state $\rho$ in the unit time is given by a quantum channel $\Phi$, that is a completely positive trace-preserving map on $\mathcal{B}(\mathcal{H})$, the algebra of bounded operators on $\mathcal{H}$~\cite{heinosaari2011mathematical}. The adjoint map $\Phi^{\dagger}$ of the channel $\Phi$ is defined with respect to the Hilbert-Schmidt scalar product $\braket{A}{B}_{\mathrm {HS}}=\Tr(A^{\dagger}B)$ via
			\begin{equation}
				\braket{A}{\Phi(B)}_{\mathrm {HS}}= \braket*{\Phi^{\dagger}(A)}{B}_{\mathrm {HS}},\quad A,B \in \mathcal{B}(\mathcal{H}).
			\end{equation}
			$\Phi^\dagger$ is a completely positive unital map on $\mathcal{B}(\mathcal{H})$, describing the evolution of system observables, i.e.\ the dynamics in the Heisenberg picture.
			
			Working in the Schr\"odinger picture, given a quantum channel $\Phi$, the evolution at time $t=n\in \mathbb{N}$ will be given by the $n$-fold composition $\Phi^{n}=\Phi \circ \Phi\circ \dots \circ \Phi$, so we obtain a discrete-time semigroup $(\Phi^n)_{n\in\mathbb{N}}$. 
			
		The spectrum $\spec(\Phi)$ of a quantum channel $\Phi$ in finite dimensions is the set of its eigenvalues, and it satisfies the following three properties
		\begin{itemize}
				\item $1 \in \spec(\Phi)$,
				\item $\lambda \in \spec(\Phi) \Rightarrow \bar{\lambda}\in \spec(\Phi)$,
				\item $\spec(\Phi)\subseteq \{ \lambda \in \mathbb{C} \,| \, |\lambda| \leqslant 1 \}$.
			\end{itemize}
			Incidentally, these properties still hold for the larger class of positive and trace-preserving maps~\cite{wolf2012quantum}.
			In particular, the \textit{peripheral spectrum} $\spec_P(\Phi)$ of the channel $\Phi$ is defined as the set of eigenvalues with unit modulus, namely 
			\begin{equation}
				\spec_{P}(\Phi)= \{ \lambda \in \spec(\Phi) \,|\, |\lambda|=1 \}.
			\end{equation}
			The asymptotic dynamics, obtained in the limit $n\rightarrow \infty$, takes place inside the \textit{asymptotic}, \textit{peripheral}, or \textit{attractor subspace} of $\Phi$, defined as
			\begin{equation}
				\Attr(\Phi)= \bigoplus_{\lambda \in \spec_{P}(\Phi)} \Ker(\Phi - \lambda\mathsf{1}).
			\end{equation}
			Equivalently,
			\begin{equation}
				\Attr(\Phi)= \mbox{span} \{ X \,|\, \Phi(X)=\lambda X \mbox{ for some $\lambda\in\spec_{P}(\Phi)$} \} ,
			\end{equation}
			i.e.\ the space spanned by the eigenvectors corresponding to the peripheral eigenvalues. 
			
			Physically, the attractor subspace is the space of recurrences $Y$ of the map $\Phi$, namely~\cite{wolf2012quantum}
			\begin{equation}
				\forall \epsilon >0\,\, \exists n \in \mathbb{N} \mbox{  s.t.  } \|\Phi^{n}(Y)-Y\|_{\mathrm{HS}}\leqslant \epsilon.
			\end{equation}
			Some of its elements are limit cycles
			\begin{equation}
				\Phi^{n}(Y)=Y \mbox{ for some }n \in \mathbb{N},
			\end{equation}
			and, in particular, fixed points
			\begin{equation}
				\Fix(\Phi)=\{ Y \,|\, \Phi(Y)=Y \}.
			\end{equation}
			
			Now, let us introduce the spectral decomposition of a quantum channel $\Phi$
			\begin{equation}
				\Phi=\sum_{k=1}^{N} (\lambda_k\mathcal{P}_k + \mathcal{N}_k),
				\label{Jordan_dec}
			\end{equation}
			where $\mathcal{P}_{k}$ and $\mathcal{N}_{k}$ are the eigenprojections and eigenilpotents of $\Phi$ corresponding to the $k$-th eigenvalue $\lambda_{k}$, $k=1,\dots , N$~\cite{kato2013perturbation}. In particular, let us denote by $\mathcal{P}$ the eigenprojection onto the fixed point space $\Fix(\Phi)$, corresponding to the eigenvalue $\lambda=1$. Also, let $\mathcal{P}_P$ be the \textit{peripheral eigenprojection}, i.e.\ the one onto the attractor subspace $\Attr(\Phi)$, which may be written explicitly as
			\begin{equation}
				\mathcal{P}_{P}=\sum_{\lambda_k\in\spec_{P}(\Phi)}\mathcal{P}_k.
			\end{equation}
			Both the eigenprojections $\mathcal{P}$ and $\mathcal{P}_P$ can be written in terms of the channel $\Phi$ as
			\begin{align}
				&\mathcal{P}=\lim_{N\rightarrow \infty} \frac{1}{N}\sum_{n=1}^{N}\Phi^n,\label{P_for}\\
				&\mathcal{P}_P=\lim_{i\rightarrow \infty}\Phi^{n_i},\label{P_p_for}
			\end{align}
			for some increasing subsequence $( n_i)_{i\in\mathbb{N}}$. 
			From Eqs~\eqref{P_for} and~\eqref{P_p_for} it is clear that both $\mathcal{P}$ and $\mathcal{P}_P$ are quantum channels.
			We can also introduce the \textit{peripheral channel} $\Phi_P$ of $\Phi$:
			\begin{equation}
				\Phi_{P}=\mathcal{P}_P\Phi=\Phi\mathcal{P}_P=\sum_{\lambda
					_{k}\in\spec_{P}(\Phi)}\lambda_k\mathcal{P}_k,
				\label{per_ch}
			\end{equation}
			Note that in Eq.~\eqref{per_ch} the eigenilpotents do not appear, since the peripheral eigenvalues are semisimple~\cite{wolf2012quantum}.
			
			Interestingly, one gets
			\begin{equation}
				\Fix(\Phi) \subseteq 0\oplus\mathcal{B}(\mathcal{H}_{0}),
				\label{fix_DEC}
			\end{equation}
			where $ \mathcal{H}_{0} = \supp \mathcal{P}(\mathbb{I})=\Ran \mathcal{P}(\mathbb{I})$~\cite[Proposition 6.9]{wolf2012quantum}. Here, $\supp A$ and $\Ran A$ denote the support and the range of the operator $A$, and $0$ acts on $\mathcal{H}_{0}^{ \perp}$, the orthogonal complement  of $\mathcal{H}_{0}$.
			Therefore 
			\begin{equation}
				X \in \Fix(\Phi) \quad \Rightarrow\quad \supp(X),\quad \Ran(X) \subseteq \mathcal{H}_{0},
			\end{equation}
			and, in particular, $\mathcal{P}(\mathbb{I})$ is a maximum-rank fixed point of $\Phi$.
			By applying~\cite[Proposition 6.9]{wolf2012quantum} to $\mathcal{P}_{P}$, we obtain
			\begin{equation}
				\Attr(\Phi)\subseteq 0\oplus\mathcal{B}(\mathcal{H}_{0}^{\prime}) ,
				\label{Attr_struc}
			\end{equation} 
			where $\mathcal{H}_{0}^{\prime}=\supp\mathcal{P}_{P}(\mathbb{I})=\Ran \mathcal{P}_{P}(\mathbb{I}) \supseteq \mathcal{H}_{0}$. In fact, as proved in~\ref{SUPPPSUPPP} and~\ref{PPNEQP}, it turns out that
			\begin{equation}
				\mathcal{H}_{0}=\mathcal{H}_{0}^{\prime},
				\label{suppPpsuppP}
			\end{equation}
			even though, in general,
			\begin{equation}
				\mathcal{P}_{P}(\mathbb{I}) \neq \mathcal{P}(\mathbb{I}).
				\label{PpneqP}
			\end{equation}
			Finally, notice that $\mathcal{P}_{P}(\mathbb{I})=\mathcal{P}(\mathbb{I})$  for unital channels, i.e.\ $\Phi(\mathbb{I})=\mathbb{I}$, or channels with trivial peripheral spectrum, viz. $\spec_{P}(\Phi)=\{ 1 \}$, such as primitive maps (see Definition~\ref{prim_def}). 
			
		\end{subsection}
		
		\begin{subsection}{Infinitely divisible channels}
			\label{inf_ch}
			We now introduce the important class of infinitely divisible quantum channels which in the following will be proved to exhibit a special asymptotic behavior.
			
			Let $(\Phi_t)_{t\in \mathbb{R}^{+}}$ be a one-parameter family of channels with $\Phi_0=\mathsf{1}$, the identity channel. If the semigroup property $\Phi_{t_1+t_2}=\Phi_{t_1}\Phi_{t_2}$ is satisfied for all $t_1,t_2\in\mathbb{R}^+$, the family is called a quantum dynamical semigroup, and it is of the form
			\begin{equation}
				\Phi_t=e^{t\mathcal{L}},\quad t\in\mathbb{R}^+.
			\end{equation}
			Here, $\mathcal{L}$ is a GKLS generator and has the following structure~\cite{GKS_76,Lindblad_76}:
			\begin{equation}
				\label{GKLS}
				\mathcal{L}(X)=-i[H, X]+\sum_{k=1}^{d^{2}-1}
				\Bigl( A_{k} X A_{k}^{\dagger}-\frac{1}{2} \{ A_{k}^{\dagger}A_{k} , X \} \Bigr) =\mathcal{L}_{H}(X)+\mathcal{L}_{D}(X),
			\end{equation}
			for $X \in \mathcal{B}(\mathcal{H})$.
			Here, the square (curly) brackets denote the (anti)commutator, $H=H^\dagger$ is the system Hamiltonian, and the noise operators $A_{k}$ are arbitrary operators on a $d$-dimensional Hilbert space.
			Moreover,
			$\mathcal{L}_{H}$ and $\mathcal{L}_{D}$ stand for the Hamiltonian and dissipative parts of the generator $\mathcal{L}$, respectively. A quantum channel $\Phi=e^{\mathcal{L}}$, with $\mathcal{L}$ of the form~\eqref{GKLS}, is said to be Markovian. Note that a Markovian channel $\Phi$ is invertible, which is equivalent to say that $0\not\in\spec(\Phi)$. 
			
			Another relevant class of channels are the idempotent channels, satisfying the condition $\Phi^2=\Phi$. A paradigmatic example is the contraction channel
			\begin{equation}
				\Phi(X)=\tr(X)\rho,\quad \rho\in\mathcal{S}(\mathcal{H}), \quad X\in\mathcal{B}(\mathcal{H}).
				\label{contr_ch}
			\end{equation} 
			Unlike Markovian channels, idempotent channels are not invertible, except for the trivial identity channel $\mathsf{1}$. However, both classes of channels share the following more general property, called infinite divisibility.
			\begin{definition}
				\label{def:ifchannel}
				Let $\Phi$ be a quantum channel. Then $\Phi$ is said to be infinitely divisible iff for all $n\in\mathbb{N}$,
				\begin{equation}
					\Phi=\Phi_n^n, 
				\end{equation}
				for some quantum channel $\Phi_n$, called an $n$-th root of $\Phi$.
			\end{definition}
			While Markovian and idempotent channels are infinitely divisible, there are examples of infinitely divisible channels which are neither Markovian nor idempotent. Nevertheless, it can be proved that $\Phi$ is infinitely divisible iff
			\begin{equation}
				\Phi=\mathcal{Q}e^\mathcal{L},
			\end{equation} 
			where $\mathcal{Q}$ is an idempotent channel, $\mathcal{L}$ a GKLS generator, and $\mathcal{Q}\mathcal{L}=\mathcal{Q}\mathcal{L}\mathcal{Q}$~\cite{denisov1988infinitely,wolf2008dividing}.
			
		\end{subsection}
		
	\begin{subsection}{The structure of the asymptotic map}
			\label{struc_wolf}
			One may ask whether there exists a finer structure for the asymptotic manifold $\Attr(\Phi)$ and for the action of the \textit{asymptotic map} 
			\begin{equation}
				\hat{\Phi}_P=\Phi\vert_{\Attr(\Phi)}=\Phi_{P}\vert_{\Attr(\Phi)}
			\end{equation}
			of $\Phi$. The answer is given by the following  structure theorem of Perez-Garcia and Wolf~\cite[Theorem 8]{wolf2010inverse}.
			\begin{theorem}[Asymptotic map]
				\label{Wolf_struc}
				Let $\Phi$ be a quantum channel on $\mathcal{B}(\mathcal{H})$, and $\mathcal{P}$ its eigenprojection corresponding to $\lambda=1$.
				
				There exists a decomposition of the Hilbert space $\mathcal{H}$ 
				\begin{equation}
					\mathcal{H}=\mathcal{H}_{0}^\perp \oplus \bigoplus_{k=1}^{M} \mathcal{H}_{k,1}\otimes \mathcal{H}_{k,2},
					\label{H_dec}
				\end{equation}
				with $\mathcal{H}_0=\supp\mathcal{P}(\mathbb{I})$, $\mathcal{H}_{k,i}$ being Hilbert spaces ($k=1,\dots, M$ and $i=1,2$), and there exist positive definite density operators $\rho_{k}$ on $\mathcal{H}_{k,2}$ such that
				\begin{equation}
					\Attr(\Phi)=0\oplus \bigoplus_{k=1}^{M}\mathcal{B}(\mathcal{H}_{k,1}) \otimes \mathbb{C}\rho_{k}.
					\label{attr_struc}
				\end{equation}
				 
				Therefore every $X\in \Attr(\Phi)$ can be written as
				\begin{equation}
					\label{attr_struct_element}
					X=0\oplus \bigoplus_{k=1}^{M} x_{k} \otimes \rho_{k},
				\end{equation}
				for some operators $x_{k}\in\mathcal{B}(\mathcal{H}_{k,1})$.
				
				Moreover, there exist unitaries $U_{k}$ on $\mathcal{H}_{k,1}$ and a permutation $\pi$ on the set $\{ 1,\dots , M \}$ such that the asymptotic map of the quantum channel $\Phi$ reads
				\begin{equation}
					\hat{\Phi}_P(X)=0 \oplus \bigoplus_{k=1}^{M} U_{k}x_{\pi(k)}U_{k}^{\dagger}\otimes \rho_{k},
					\label{struc_phiP}
				\end{equation}
				with $X \in \Attr(\Phi)$ in the form~\eqref{attr_struct_element}.
			\end{theorem}
			
			\begin{remark}\label{attr_inva}
				The permutation $\pi$ acts on subsets of $\{ 1,\dots , M \}$, so that the corresponding Hilbert spaces must have the same dimension, i.e.
				\begin{equation}
					d_{\pi(k)}=d_k,\quad k=1,\dots , M,
				\end{equation}
				with $d_{k}=\dim(\mathcal{H}_{k,1})$,
				consistently with the fact that $\Attr(\Phi)$ is a proper invariant subspace for $\Phi$,
				\begin{equation}
					\Phi\Attr(\Phi)=\Attr(\Phi).
					\label{Attr_inva}
				\end{equation}  
				\label{per_eq_dim}
			\end{remark}
			\begin{remark}[\cite{AFK_asympt}]
				If $\Phi$ is faithful, i.e.\ it has a full-rank fixed state, then $\dim(\mathcal{H}_{0}^{ \perp})=0$, so that the zero term in~\eqref{attr_struc} disappears. Consequently,
				\begin{equation}
					\mathcal{P}_P(X)= \sum_{k=1}^{M} \tr_{k,2}(P_{k}XP_{k})\otimes \rho_{k},
					\label{P_p_struc_faith}
				\end{equation}
				
				where $P_{k}$ is the projection onto $\mathcal{H}_{k,1}\otimes \mathcal{H}_{k,2}$ and $\tr_{k,2}$ denotes the partial trace over $\mathcal{H}_{k,2}$.

				Moreover, it can be shown that $\Attr(\Phi^{\dagger})$ is a unital $\ast$-algebra~\cite{davidson1996c} with the following structure~\cite[Theorem 1]{carbone2020period}
				\begin{equation}
					\Attr(\Phi^{\dagger})=\bigoplus_{k=1}^{M} \mathcal{B}(\mathcal{H}_{k,1}) \otimes \mathbb{I}_{k,2},
				\end{equation}
				where $\mathbb{I}_{k,2}$ is the identity on $\mathcal{H}_{k,2}$. 	
				Following the terminology of~\cite{Viola_inform_pres}, $\Attr(\Phi)$ is called the distorted algebra of $\Attr(\Phi^{\dagger})$ in the sense that it is an algebra with respect to the modified product
				\begin{equation}
					A \star B := A\mathcal{P}_P(\mathbb{I})^{-1}B \in \Attr(\Phi),\quad A,B\in \Attr(\Phi).
					\label{mod_prod}
				\end{equation}
				This can be immediately seen if we observe that
				\eqref{P_p_struc_faith} implies that
				\begin{equation}
					\mathcal{P}_P(\mathbb{I})= \bigoplus_{k=1}^{M}m_k\mathbb{I}_{k,1}\otimes \rho_{k},
				\end{equation}
				where $\mathbb{I}_{k,1}$ is the identity on $\mathcal{H}_{k,1}$ and $m_k=\dim(\mathcal{H}_{k,2})$.
				\label{faith}
			\end{remark}
			
			\begin{remark}
				As a consequence of~\eqref{struc_phiP}, every maximum-rank fixed state $\sigma$ has the form
				\begin{equation}
					\sigma =\frac{1}{\sum_{k=1}^{M}\Tr(\sigma_{k})} \Bigl(0 \oplus \bigoplus_{k=1}^{M}\sigma_{k}\otimes \rho_{k}\Bigr),
					\label{max_rank_dec}
				\end{equation}
				with $\sigma_k > 0$ such that
				\begin{equation}
					U_k\sigma_{\pi(k)}U_k^\dagger =\sigma_k,\quad k=1,\dots , M.
					\label{fix_con}
				\end{equation}	
			\end{remark}
			
		\end{subsection}
	\end{section}
	
	\begin{section}{Cyclic structure of the asymptotic map}
		\label{aternativeexpression}
		
		In this Section, we are going to exhibit an interesting alternative form for the  asymptotic map $\hat{\Phi}_P$ of a quantum channel $\Phi$. In particular, we are going to show that, in each disjoint cycle of the permutation $\pi$, the channel can be written as the tensor product between a unitary channel and a permutation between the density matrices $\rho_k$ in the decomposition~\eqref{attr_struc} of the asymptotic manifold. In order to do this, we will first need to rewrite the map~\eqref{struc_phiP} in a form where in  each cycle all the unitaries $U_k$ are the same.
		
		According to Theorem~\ref{Wolf_struc}, given the Hilbert space decomposition~\eqref{H_dec} and an element $X\in \Attr(\Phi)$ in the form~\eqref{attr_struct_element}, the action of the asymptotic map $\hat{\Phi}_P$ on $X$ is in the form~\eqref{struc_phiP}. 
		Notice that local unitary transformations (i.e.\ local basis changes) $V_k : \mathcal{H}_{k,1} \to \mathcal{H}_{k,1}'$, with $k=1,\dots,M$, induce a unitary map $V:\mathcal{H}\to\mathcal{H}'$, where 
		\begin{equation}
			\mathcal{H}'=\mathcal{H}_{0}^\perp \oplus \bigoplus_{k=1}^{M} \mathcal{H}_{k,1}'\otimes \mathcal{H}_{k,2}
		\end{equation}
		has exactly the same structure of $\mathcal{H}$. We will use this freedom in order to simplify the expression of the asymptotic  map $\hat{\Phi}_P$.
		
		Any operator of the asymptotic manifold $X\in \Attr(\Phi)$ is transformed by $V$ into
		\begin{equation}
			X'=VXV^\dag=0\oplus\bigoplus_{k=1}^{M} x'_k\otimes \rho_k,
		\end{equation}
		where
		\begin{align}
			x'_k=V_k x_k V_k^\dagger.
		\end{align}
		
		Correspondingly, the action of the asymptotic map on $X'$ becomes
		\begin{equation}
			\hat{\Phi}_P'(X')= V \hat{\Phi}_P(X) V^\dag = 
			V \hat{\Phi}_P(V^\dag X' V) V^\dag,
		\end{equation}
		that is
		\begin{equation}
			\hat{\Phi}_P'(X')=0\oplus\bigoplus_{k=1}^M (V_k U_k V_{\pi(k)}^\dagger )x'_{\pi(k)}( V_{\pi(k)} U_k^\dagger V_k^\dagger)\otimes\rho_k=0\oplus \bigoplus_{k=1}^MU'_k x'_{\pi(k)} U_k'^\dagger\otimes \rho_k.
		\end{equation}
		As a result, changing the local basis will change the operators $U_k$ as
		\begin{equation}\label{eq:changingbasis}
			U_k\mapsto U'_k= V_kU_kV_{\pi(k)}^\dagger.
		\end{equation}

		In general a permutation $\pi$ is the product of disjoint cycles.
		Assume first that $\pi=(k_1,\dots,k_M)$ is a cyclic permutation  of length $M$, that is
		\begin{equation}
			\label{cyc_per}
			\pi(k_j)=k_{j+1} \mod{M}.
		\end{equation}
		In such a case all Hilbert spaces are isomorphic,
		\begin{equation}
			\mathcal H_{1,1}'\simeq\mathcal H_{2,1}'\simeq\dots \simeq\mathcal H_{M,1}'\simeq \mathcal H_1,
		\end{equation}
		and we can write
		\begin{equation}
			\mathcal{H}'=\mathcal{H}_{0}^\perp \oplus \bigoplus_{k=1}^{M} \mathcal{H}_{1}\otimes \mathcal{H}_{k,2}.
		\end{equation}	
		
		Now, we require that the unitaries $U'_k$ be all equal :
		\begin{equation}
			U'_{k_1}=U'_{k_2}=\dots=U'_{k_M}=  U,
		\end{equation}
		which means that
		\begin{equation}
			U=V_{k_j}U_{k_j}V_{k_{j+1}}^\dagger\mod M.
		\end{equation}
		By writing $V_{k_{j+1}}$ in terms of $V_{k_{j}}$, one gets
		\begin{equation}
			V_{k_{1}}=( U^\dagger)^M V_{k_{1}} (U_{k_{1}}U_{k_{2}}\cdots U_{k_{M}}),
		\end{equation}
		or explicitly
		\begin{equation}
			U=V_{k_1} (U_{k_{1}}U_{k_{2}}\cdots U_{k_{M}})^{1/M} V_{k_1}^\dagger.
		\end{equation}
		Therefore, up to a local unitary $V_{k_1}$  on $\mathcal{H}_{k_1 ,1}\simeq\mathcal{H}_{1}$, the transformed unitary ${U}$ turns out to be a geometric mean of the original unitaries $U_{k_j}$, $j=1,\dots , M$, along the cycle.
		
	Thus, by dropping all primes, the expression~\eqref{struc_phiP} for the asymptotic map  becomes
		\begin{equation}\label{eq:Phip_uniformU}
			\hat{\Phi}_P(X)=0\oplus \bigoplus_{k=1}^M  U x_{\pi(k)} U^\dagger \otimes \rho_{k}.
		\end{equation}
		
		Moreover, we can decompose the Hilbert space $\mathcal H$ as
		\begin{equation}\label{eq:hilbertspaceuniformdecomposition}
			\mathcal H=\mathcal H_0^{\perp}\oplus  \bigoplus_{k=1}^M \mathcal H_1\otimes  \mathcal H_{k,2} = 
			\mathcal H_0^{\perp}\oplus \mathcal{H}_1 \otimes \Bigl(\bigoplus_{k=1}^M   \mathcal H_{k,2}\Bigr).
		\end{equation}
		
		With respect to the decomposition~\eqref{eq:hilbertspaceuniformdecomposition}, the attractor subspace $\Attr(\Phi)$ becomes
			
		\begin{equation}
			\mathrm{Attr}(\Phi)=0\oplus \mathcal{B}(\mathcal{H}_1) \otimes \Bigl(\bigoplus_{k=1}^{M}\mathbb{C}\rho_k\Bigr),
		\end{equation}
		and the action of $\hat{\Phi}_P$ on an element $X\in \Attr(\Phi)$,
		\begin{equation}
			X= 0 \oplus \sum_{k=1}^M x_k\otimes \hat{\rho}_k,
		\end{equation} 
		with $\hat{\rho}_k = (0\oplus 0\oplus \dots \oplus \rho_k \oplus 0\oplus\dots\oplus 0)$, is given by
		
		\begin{equation}
			\hat{\Phi}_P(X)=0\oplus \sum_{k=1}^M  U x_{\pi(k)}  U^\dagger\otimes \hat{\rho}_k
			= 0\oplus \sum_{k=1}^M  U x_{k}  U^\dagger\otimes \hat{\rho}_{\pi^{-1}(k)}.
		\end{equation}
		Now, it is easy to check that one gets 
		$\hat{\rho}_{\pi^{-1}(k)} = \Phi_\pi(\hat{\rho}_k)$, where $\Phi_\pi$ is the channel on $\mathcal{B} ( \bigoplus_{k}\mathcal{H}_{k,2} )$ acting as
		\begin{equation}
			\Phi_\pi(Z)=\bigoplus_{k=1}^M \tr( Q_{\pi(k)}Z)\rho_k,\qquad Z\in \mathcal{B} \Bigl( \bigoplus_{k=1}^{M}\mathcal{H}_{k,2} \Bigr),
		\end{equation}
		with $Q_{k}$ being the projection onto $\mathcal{H}_{k,2}$.
		
		Therefore, the asymptotic map can be written as
		\begin{equation}	
			\hat{\Phi}_P(X)	=  0\oplus (\Phi_{ U} \otimes \Phi_\pi) (X_0),\qquad X= 0 \oplus X_0 \in \Attr(\Phi),
			\label{Ch_cycl}
		\end{equation}
		with
		\begin{align}
			\Phi_{ U}(Y)={ U}Y{ U}^\dagger,\quad Y\in \mathcal{B}(\mathcal{H}_{1}),
		\end{align}
		being a unitary channel.

		The above decomposition, obtained for a permutation consisting of a single $M$-cycle, generalizes in a natural way to an arbitrary permutation with $L$ disjoint cycles $\pi=\pi_1\circ\pi_2\circ \cdots \circ\pi_L$. We have thus proved the following structure theorem for the asymptotic dynamics, which is the main result of this section.
		\begin{theorem}[Cyclic decomposition]
			\label{thm:cycles}

			Let $\Phi$ be a quantum channel on $\mathcal{B}(\mathcal{H})$.
			
			Then, there exists a 
			decomposition of the Hilbert space $\mathcal{H}$,
			\begin{equation}
				\mathcal H=\mathcal H_0^{\perp}\oplus \bigoplus_{\ell=1}^L \biggl(  \mathcal H_1^{(\ell)}\otimes \Bigl(\bigoplus_{k=1}^{m_{\ell}} \mathcal H_{k,2}^{(\ell)}\Bigr)   \biggr),
			\end{equation}
			such that the asymptotic manifold of $\Phi$ has the form
			\begin{equation}
				\Attr(\Phi)=0\oplus \bigoplus_{\ell=1}^L \biggl(  \mathcal{B}(\mathcal H_1^{(\ell)})\otimes \Bigl(\bigoplus_{k=1}^{m_{\ell}}\mathbb{C}\rho_k^{(\ell)}\Bigr)   \biggr),
				\label{Attr_cyc_struc}
			\end{equation}
			with $\rho_k^{(\ell)}$ being positive definite density operators. 
			
			The asymptotic map of $\Phi$, $\hat{\Phi}_P:\Attr(\Phi)\to\Attr(\Phi)$ is given by
			\begin{equation}
				\hat{\Phi}_P(X)=0\oplus\bigoplus_{\ell=1}^L(\Phi_{U_\ell} \otimes \Phi_{\pi_\ell}) (X_0),\quad X=0 \oplus X_0\in \Attr(\Phi),
				\label{Ch_cycl_gen}
			\end{equation}
			where 
			\begin{equation}
				\Phi_{U_\ell}(Y)={U_\ell}Y{U_\ell}^\dagger,\quad Y\in \mathcal{B}(\mathcal{H}_1^{(\ell)}),
			\end{equation} 
			are unitary channels, and
			\begin{equation}
				\Phi_{\pi_\ell}(Z)=\bigoplus_{k=1}^{m_{\ell}} \tr( Q_{{\pi_\ell}(k)}^{(\ell)} Z)\rho_k^{(\ell)},\qquad Z\in \mathcal{B} \Bigl(\bigoplus_{k=1}^{m_{\ell}} \mathcal H_{k,2}^{(\ell)}\Bigr),
			\end{equation}
			with $\pi_\ell$ being an $m_{\ell}$-cycle, and $Q_{k}^{(\ell)}$ being the projection onto $\mathcal{H}_{k,2}^{(\ell)}$, for $\ell=1,\dots,M$.
		\end{theorem}
		
	Observe that the decomposition~\eqref{Ch_cycl_gen} can also  be obtained from~\cite[Theorem 2]{guan2018structure}.
		In particular, the map $\Phi_{\pi_\ell}$ corresponds to the peripheral channel of the irreducible map $\mathcal E_\ell$ appearing there, up to the identification 
		$\Attr(\mathcal E_\ell)=\bigoplus_{k=1}^{m_{\ell}} \mathbb{C} \rho_k^{(\ell)}$, and $\pi_{\ell}$ as the permutation appearing in the expression~\eqref{struc_phiP} for the asymptotic map of $\mathcal E_\ell$.
		
		Thus, the attractor subspace reads
		\begin{equation}
			\Attr(\Phi)=0 \oplus \bigoplus_{\ell=1}^L\mathcal{B}(\mathcal H_1^{(\ell)}) \otimes \Attr(\Phi_{\pi_{\ell}}).
			\label{Ch_dec_Attr}
		\end{equation}
		Notice that, in the absence of permutations $\pi=\mathrm{id}$,~\eqref{Ch_dec_Attr} reduces to the decomposition~\eqref{struc_phiP}. A sufficient condition and a characterization for $\pi=\mathrm{id}$ will be discussed in in Sec.~\ref{per_role}. 
		
		Armed with this new expression for the asymptotic map $\hat{\Phi}_P$, we can decompose the fixed point space $\Fix(\Phi)$ of $\Phi$ as
		
		\begin{equation}
			\Fix(\Phi)= 0\oplus\bigoplus_{\ell=1}^L	\mathrm{Fix} (\Phi_{ U_\ell})\otimes \tilde\rho_\ell.
		\end{equation}
		Here
		\begin{equation}
			\tilde{\rho}_\ell=\frac{1}{m_{\ell}}\bigoplus_{k=1}^{m_{\ell}} \rho_k^{(\ell)}
		\end{equation}
		is the unique (full-rank) fixed point of the irreducible channel $\Phi_{\pi_\ell}$ and
		\begin{equation}
			\mathrm{Fix} (\Phi_{ U_\ell})=\{  U_\ell \}^\prime = \bigl\{ X \in \mathcal{B}(\mathcal{H}_1^{(\ell)}) \,:\, [X , U_\ell]=0	\bigr\},
		\end{equation}
		where the square brackets denote the commutator. 
		In addition, as already observed in~\cite[Subsection 5.3]{albert2019asymptotics}, we can obtain from~\eqref{Ch_cycl_gen} the structure of the peripheral eigenvalues and eigenvectors of $\Phi$. In particular, the structure of the peripheral spectrum is consistent with~\cite[Theorem 9]{wolf2010inverse}.  
		\end{section}
	\begin{section}{Unfolding an asymptotic map}
		\label{extens_th}
		In Section~\ref{struc_wolf}, we discussed the  expression 
		of the asymptotic map $\hat{\Phi}_P$ of a quantum channel $\Phi$, that describes the asymptotic dynamics of the quantum system under investigation. An alternative decomposition~\eqref{Attr_cyc_struc} of the attractor manifold $\Attr(\Phi)$ and its corresponding asymptotic map~\eqref{Ch_cycl_gen} was then given in Sec.~\ref{aternativeexpression}. 
		
		A fundamental question to be addressed is whether, given an arbitrary subspace $\mathcal{K}$ decomposed as in~\eqref{attr_struc} and an arbitrary map $\Phi_{\mathcal K}$ of the form~\eqref{struc_phiP} acting on it, there always exists a quantum channel on the whole algebra $\mathcal B(\mathcal H)$ whose attractor subspace is $\mathcal{K}$ and whose asymptotic map is $\Phi_{\mathcal{K}}$. The answer is affirmative, as shown in the following Theorem.
		\begin{theorem}[Unfolding an asymptotic dynamics]
			\label{th:extensiontheorem}
			Let $\mathcal H$ be a $d$-dimensional Hilbert space of the form
			\begin{equation}
				\mathcal H=\mathcal H_0^{\perp}\oplus\mathcal{H}_0, \qquad \mathcal{H}_0= \bigoplus_{k=1}^M \mathcal H_{k,1}\otimes \mathcal H_{k,2},
			\end{equation}
			with $\dim{\mathcal H_{k,1}}=d_k$.
			Let $\rho_k$ be a full-rank density operator over $\mathcal H_{k,2}$, $k=1,\dots , M$, and consider the  subspace $\mathcal{K}\subset \mathcal B(\mathcal H)$ given by
			\begin{equation}
				\mathcal{K}=0\oplus\bigoplus_{k=1}^M \mathcal{B}(\mathcal{H}_{k,1}) \otimes \mathbb{C}\rho_k.
			\end{equation}
			Let $\pi$ be a permutation of $\{ 1,\dots , M \}$ such that $d_k=d_{\pi(k)}$, and $U_k$ be unitaries in $\mathcal{B}(\mathcal{H}_{k,1})$. Consider the map $\Phi_\mathcal K:\mathcal K\rightarrow \mathcal K$ given by
			\begin{align}
				X=0\oplus\bigoplus_{k=1}^M x_k\otimes \rho_k&\mapsto \Phi_\mathcal K(X)=0\oplus\bigoplus_{k=1}^M U_kx_{\pi(k)}U_k^\dagger\otimes \rho_k.
			\end{align}
			Then, there exists a quantum channel $\Phi_{\mathrm E}$ on $\mathcal{B}(\mathcal{H})$
		such that
			\begin{itemize}
				\item[i.] $\mathrm{Attr}(\Phi_\mathrm{E})=\mathcal K$,
				\item[ii.] $(\widehat{\Phi_{E}})_P=\Phi_\mathrm{E}\vert_{\mathcal K}=\Phi_\mathcal K$.
			\end{itemize}
		\end{theorem}
		\begin{remark}
			This result is along the lines of a series of extension theorems of completely positive maps~\cite{paulsen2002completely}, such as the classic result by Arveson~\cite{arveson1969subalgebras} and the more recent theorem by Jencova~\cite{jenvcova2012generalized}. 
			It is worth mentioning that the proof of Theorem~\ref{th:extensiontheorem} is constructive and provides an explicit form for the quantum channel $\Phi_E$ which unfolds the given asymptotic dynamics. 
		\end{remark}
		\begin{proof}
			We will prove the statement by providing an explicit example of a quantum channel $\Phi_E$ satisfying the conditions  $i.$ and $ii.$. First, let us define the pinching channel		\begin{equation}
				\begin{split}
					\Phi_\mathrm {pinch}:\mathcal B(\mathcal H)&\rightarrow \mathcal B(\mathcal H_0^\perp)\oplus\bigoplus_{k=1}^M \mathcal B(\mathcal H_{k,1}\otimes\mathcal H_{k,2}),\\
					Z& \mapsto Q^\perp Z Q^\perp\oplus \bigoplus_{k=1}^M P_k{Z}P_k,
				\end{split}
			\end{equation}
			where $P_k$ are the projections onto $\mathcal H_{k,1}\otimes \mathcal H_{k,2}$, while $Q^\perp$ is the projection onto $\mathcal H_0^\perp$. Clearly, $\Phi_\mathrm {pinch}$ is a quantum channel. 
			
			Let us post-compose $\Phi_\mathrm {pinch}$ with the sum of two maps. The first one is a variant of the one-state contraction channel 
			\begin{equation}
				\begin{split}
					\Phi_0^\perp:\mathcal B(\mathcal H_0^\perp)\oplus\bigoplus_{k=1}^M \mathcal B(\mathcal H_{k,1}\otimes\mathcal H_{k,2})&\rightarrow 0\oplus\mathcal B(\mathcal H_0),\label{Phi_0_sp}\\
					Z_0\oplus\bigoplus_{k=1}^M Z_k&\mapsto  0\oplus \tr(Z_0) \sigma,
				\end{split}
			\end{equation}
			with $\sigma\in\mathcal S(\mathcal H_0)$. 
			It is immediate to check that $\Phi_0^\perp$ is a nilpotent completely positive map of degree 2.
			[More generally, we could replace $\tr(Z_0)\sigma$ with $\Theta(Z_0)$, with $\Theta:\mathcal B(\mathcal H_0^\perp)\rightarrow B(\mathcal H_0)$ being an arbitrary quantum channel].

			The second map has the form
			\begin{equation}
				\begin{split}
					\Phi_0:\mathcal B(\mathcal H_0^\perp)\oplus\bigoplus_{k=1}^M  \mathcal B(\mathcal H_{k,1}\otimes\mathcal H_{k,2})&\rightarrow 0\oplus\bigoplus_{k=1}^M \mathcal B(\mathcal H_{k,1}\otimes\mathcal H_{k,2}),\\
					Z_0\oplus\bigoplus_{k=1}^M Z_k&\mapsto 0\oplus  \bigoplus_{k=1}^M U_{k}\tr_{\pi(k),2}(Z_{\pi(k)})U_{k}^{\dagger}\otimes \rho_k,
				\end{split}
				\label{eq:Phi0}
			\end{equation}
			where $\tr_{\pi(k),2}$ is the partial trace over $\mathcal{H}_{\pi(k),2}$.
			
			We claim that the map 
			\begin{equation}\label{eq:explicitPhi_E}
				\Phi_\mathrm E=(\Phi_0+\Phi_0^\perp)\circ \Phi_\mathrm {pinch},
			\end{equation}
			is indeed the sought extension.
			
			We need to check that $\Phi_\mathrm E$ is a quantum channel, and that conditions $\textit i.$ and $\textit ii.$ hold.
			Clearly, $\Phi_\mathrm E$ is trace-preserving by construction. For what concerns complete positivity, we only need to prove that $\Phi_0$ satisfies this property. The proof  is straightforward but lengthy, and it is given in~\ref{CP_Phi_0}.
			
			Furthermore, observe that
			\begin{equation}
				(\Phi_\mathrm E)^n(\mathcal{B}(\mathcal{H}))=\mathcal K, \quad \forall n\geqslant 2,
				\label{per_Phi_E}
			\end{equation}
			implying $\mathrm{Attr}(\Phi_\mathrm E)=\mathcal K$ by~\eqref{P_p_for}. Condition $ii.$ follows easily from the definition~\eqref{eq:explicitPhi_E} of $\Phi_\mathrm{E}$.
		\end{proof}
		\begin{remark}
			Observe that, if $\mathcal H_0^\perp=0$, then we have $\Phi_0^\perp=0$ by~\eqref{Phi_0_sp} and the map~\eqref{eq:explicitPhi_E} reduces to the faithful channel $\Phi_\mathrm{E}=\Phi_0\circ\Phi_\mathrm{pinch}$, coinciding with its peripheral channel. Incidentally, this channel appears in the proof of Theorem 9 of~\cite{wolf2010inverse} in the particular case $m_k=1$, $k=1,\dots , M$, i.e. $\rho_k=1$. 
			
		\end{remark}
		
	\end{section}
	
	\begin{section}{On the role of permutations}
		\label{per_role}
		One of the most intriguing features of the asymptotic dynamics~\eqref{struc_phiP} is the occurrence of permutations between the factors of the decomposition~\eqref{H_dec}, which in general prevents the asymptotic map $\hat{\Phi}_P$ to be a unitary channel and, indeed, a Hilbert-Schmidt unitary. For the sake of clarity, we give the following definition.
		
	\begin{definition}
			Let $\Phi$ be a quantum channel. Then, we say that its asymptotic map
			\begin{equation}
				\hat \Phi_P:\Attr(\Phi)\rightarrow \Attr(\Phi)
			\end{equation}	
			is unitary iff
			\begin{equation}
				\hat{\Phi}_P(X)=UXU^\dagger,
				\label{unit_asy}
			\end{equation}
			with $X\in\Attr(\Phi)$, and $U\in\mathcal{B}(\mathcal{H})$ being unitary, $UU^\dagger=U^\dagger U=\mathbb{I}$.
		\end{definition}
		
		Obviously,
		\begin{equation}
			\pi=\mathrm{id}\quad \Rightarrow \quad  \hat{\Phi}_P(X)=UX U^\dagger \quad\Rightarrow\quad \hat{\Phi}_{P}\hat{\Phi}_{P}^{\dagger}=\hat{\Phi}_{P}^{\dagger}\hat{\Phi}_{P}=\mathsf{1}_{P},
			\label{HS-unit}
		\end{equation}
		where $\mathsf{1}_P$ is the identity on $\Attr(\Phi)$. 
		When the last equality in~\eqref{HS-unit} holds, we say that the asymptotic map is Hilbert-Schmidt unitary. Here, the adjoint $\hat{\Phi}_{P}^{\dagger}$ is defined via the Hilbert-Schmidt inner product,
		\begin{equation}
			\braket*{X}{\hat{\Phi}_P(Y)}_{\mathrm {HS}}=\braket*{\hat{\Phi}_{P}^{\dagger}(X)}{Y}_{\mathrm {HS}},\quad X,Y \in \Attr(\Phi),
		\end{equation}
		and can be easily shown to read 
		\begin{equation}
			\hat{\Phi}_{P}^{\dagger}(X)=0 \oplus \bigoplus_{k=1}^{M}U_{\pi^{-1}(k)}^{\dagger}x_{\pi^{-1}(k)}U_{\pi^{-1}(k)} \otimes \frac{\tr(\rho_{\pi^{-1}(k)}^{2})}{\tr(\rho_{k}^{2})} \rho_{k},\quad X\in \Attr(\Phi).
			\label{phiPd_expr}
		\end{equation}
		
		In words, the absence of permutations implies that the asymptotic map is unitary and thus Hilbert-Schmidt unitary.

		However, notice that the converse of the implications~\eqref{HS-unit} does not hold in general. Indeed, consider the qubit channel~\cite[Section 3.1]{albert2019asymptotics}:
		\begin{equation}
			\Phi(X)=\frac{1}{2}(\sigma_{1} X \sigma_{1}+\sigma_{3} X \sigma_{3}),\quad X\in\mathcal{B}(\mathbb{C}^2),
			\label{unit_count}
		\end{equation}
		where $\sigma_j,\,j=1,2,3$ are the Pauli matrices.
		The peripheral eigenvalues and eigenvectors read 
		\begin{align}
			\Phi^\dagger(\mathbb{I})=\Phi(\mathbb{I})=\mathbb{I},\qquad \Phi^{\dagger}(\sigma_{2})=\Phi(\sigma_{2})=-\sigma_{2},
		\end{align}
		so that
		\begin{equation}
			\Attr(\Phi)=\{ \alpha \mathbb{I}+\beta \sigma_{2} \,|\, \alpha,\beta\in\mathbb{C} \}.
		\end{equation}
		In the basis in which $\sigma_{2}$ is diagonal, $X\in \mathrm{Attr}(\Phi)$ has the form
		\begin{equation}
			X= \begin{pmatrix}\alpha + \beta & 0\\
				0 & \alpha - \beta\end{pmatrix},
		\end{equation}
		and the action of the asymptotic map $\hat{\Phi}_P$ on it will be given by
		\begin{equation}
			\hat{\Phi}_{P}(X)=\begin{pmatrix}\alpha - \beta & 0\\
				0 & \alpha + \beta\end{pmatrix}.
		\end{equation}
		Thus, $\hat{\Phi}_{P}$ is just the flip of the two diagonal coefficients of $X$. However, $\hat{\Phi}_{P}$ is a unitary channel and, in particular,
		\begin{equation}
			\hat{\Phi}_{P}(X)=\sigma_{3}X\sigma_{3},\quad X\in\Attr(\Phi).
		\end{equation}
		So, to sum up, 
		\begin{equation}
			\hat{\Phi}_{P}(X)=UX U^{\dagger} \quad \not\Rightarrow\quad  \pi=\mathrm{id},
			\label{unit_per}
		\end{equation}
		implying that
		\begin{equation}
			\hat{\Phi}_{P}^{\dagger}=\hat{\Phi}_{P}^{-1} \quad \not\Rightarrow\quad  \pi=\mathrm{id}.
		\end{equation}
		We  conclude that Eq.~\eqref{unit_asy} and, consequently, the Hilbert-Schmidt unitarity of $\hat{\Phi}_P$ are not sufficient conditions for the absence of permutations. Instead, a characterization is given in the following Lemma.
		\begin{lemma}
			Let $\Phi$ be a quantum channel, with asymptotic map $\hat{\Phi}_P$. Then $\hat{\Phi}_P$ has no permutations iff, for all $n\in\mathbb N$,  $\hat \Phi_P$ is the $n$-fold composition of the asymptotic map of some quantum channel $\Phi_n$. Explicitly
			\begin{equation}
				\pi=\mathrm{id} \quad \text{iff}\quad\hat{\Phi}_P = (\hat{\Phi}_{n,P})^{n}\quad \forall n \in\mathbb{N} ,
			\end{equation}
			where $\hat{\Phi}_{n,P}$ is the asymptotic map of a quantum channel $\Phi_n$ with attractor subspace $\Attr(\Phi_n)=\Attr(\Phi)$.
			\label{ch_pi_id}
		\end{lemma}
		\begin{proof} The $\Rightarrow$ implication is immediate, since we can simply set $\forall n\in\mathbb{N}$
			\begin{equation}
				\hat{\Phi}_{n,P}(X)= 0 \oplus \bigoplus_{k=1}^M U_k^{1/n}x_k (U_k^{1/n})^\dagger \otimes \rho_k,\quad X\in\Attr(\Phi),
			\end{equation} 
			which can be thought as the asymptotic maps of channels $\Phi_n$ constructed through the extension Theorem~\ref{th:extensiontheorem}.
			
			With regards to the other implication $\Leftarrow$, we can explicitly write the maps $\hat{\Phi}_{n,P}$ as
			\begin{equation}
				\hat{\Phi}_{n,P}(X)=0 \oplus \bigoplus_{k=1}^M U_{k,n}x_{\pi_{n}(k)}U_{k,n}^\dagger \otimes \rho_k,\quad n\in\mathbb{N},
			\end{equation} 
			for some unitaries $U_{k,n}$ and permutations $\pi_n$ on $\{ 1,\dots , M\}$.
			Therefore, our assumption reads
			\begin{equation}
				U_{k,n}U_{\pi_{n}(k),n}\dots U_{\pi_{n}^{n-1}(k),n}x_{\pi_{n}^{n}(k)} U_{\pi_{n}^{n-1}(k),n}^\dagger \dots U_{\pi_{n}(k),n}^\dagger U_{k,n}^\dagger=U_k x_{\pi(k)} U_k^\dagger,
			\end{equation}
			for $k=1,\dots,M,\, n\in\mathbb{N}$. In other words, 
			\begin{equation}
				\widetilde{U}_{k,n}x_{\pi_{n}^n(k)}\widetilde{U}_{k,n}^\dagger= x_{\pi(k)},\quad k=1,\dots,M,\, \quad n\in\mathbb{N},
				\label{eq_sys}
			\end{equation}
			where
			\begin{equation}
				\widetilde{U}_{k,n}=U_k^\dagger \prod_{l=0}^{n-1}U_{\pi_{n}^l(k),n},
			\end{equation}
			with $\pi_n^0=\mathrm{id}$. 
			
			We claim that this implies that $\pi=\pi_n^n$ for all $n$. Indeed, if there is a $\bar{n}\in\mathbb{N}$ such that $\pi \neq \pi_{\bar{n}}^{\bar{n}}$, the set of equations~\eqref{eq_sys} is not  satisfied for all $x_k \in \mathcal{B}(\mathcal{H}_{k,1})$ when $n=\bar{n}$. 
			
			Therefore, by taking $n$ to be the least common multiple of $1,\dots,M$, we get $\pi=\mathrm{id}$.
		\end{proof}
		
		Now, we  use the previous lemma to prove the main result of this Section.
		\begin{theorem}
			Let $\Phi$ be an infinitely divisible quantum channel. Then its asymptotic dynamics has no permutations, $\pi=\mathrm{id}$, and thus is unitary.
			\label{inf_SUFF}
		\end{theorem}
		\begin{proof} According to Lemma~\ref{ch_pi_id} it is sufficient to prove that
			\begin{equation}
				\hat{\Phi}_P = (\hat{\Phi}_{n,P})^{n},\quad \forall n \in\mathbb{N} ,
				\label{ch_no_per}
			\end{equation}
			where $\Phi_n^n=\Phi$, as in Definition~\ref{def:ifchannel}. First, observe that 
			\begin{align}
				&\Attr(\Phi)=\Attr(\Phi_n),\label{Attr_nth_root}\\
				&\spec(\Phi)=\spec(\Phi_n)^n=\{\lambda^n\,|\,\lambda\in\spec(\Phi_n)\},\label{eigen_nth_root}
			\end{align}
			as it can be seen from the Jordan decompositions of $\Phi$ and $\Phi_n$ (see~\ref{NTH-ROOT}). So, if $\mathcal{P}_{P,n}$ is the peripheral eigenprojection of $\Phi_n$, we can compute
			\begin{equation}
				\Phi_P=\mathcal{P}_P\Phi_n^n=\mathcal{P}_{P,n}\Phi_n^{n}=(\mathcal{P}_{P,n}\Phi_n)^n=\Phi_{n,P}^n,
				\label{phi_P_div}
			\end{equation} 
			where we used $[\mathcal{P}_{P,n} , \Phi_n]=0$, by~\eqref{P_p_for}. The statement follows by restricting both sides of~\eqref{phi_P_div} to the attractor subspace $\Attr(\Phi)=\Attr(\Phi_n)$.
		\end{proof}

		An immediate corollary to Theorem~\ref{inf_SUFF} is the following~\cite{AFK_asympt}.
		\begin{corollary}
			Let $\Phi$ be a quantum channel with asymptotic map $\hat{\Phi}_{P}$.
			\begin{enumerate}
				\item If $\Phi=e^{\mathcal{L}}$, with $\mathcal{L}$ being a GKLS generator (Markovian channel), then $\hat{\Phi}_{P}(X)=U X U^{\dagger}$, with $U$ unitary,
				\item If $\Phi^{2}=\Phi$ (idempotent channel), then $\hat{\Phi}_{P}(X)=U X U^{\dagger}$, with $U$ unitary.
			\end{enumerate} 
			\label{suff_unit_cor}
		\end{corollary}
		The first implication is consistent with the dynamics at large times of quantum dynamical semigroups (Section IV of~\cite{albert2014symmetries}), marking a difference with the discrete-time semigroups at the level of the asymptotics (see also Section 5.2 of~\cite{baumgartner2012structures}). Ultimately, the absence of permutations in the expression~\eqref{struc_phiP} for $\hat{\Phi}_P$ is linked to the divisibility of the channel. 
		
		The above results are also consistent with~\cite[Section 4]{AFK_asympt} where it was shown that the unitary Tomita-Takesaki modular dynamics on the attractor subspace $\Attr(\Phi)$ is the asymptotic map of $\Phi^M$, with $M$ denoting the least common multiple of the lenghts of the cycles of the permutation. In such way, the original discrete-time dynamics $(\Phi^n)_{n\in\mathbb{N}}$ associated with $\Phi$ was coarse-grained, becoming ``sufficiently" Markovian for the disappearance of permutations.    
		
		There are other classes of channels $\Phi$ for which no permutations appear in their asymptotic dynamics $\hat{\Phi}_P$. For example, this is true for primitive channels, whose definition is recalled here, together with that of irreducible channels~\cite{wolf2012quantum}.
		\begin{definition}
			Let $\Phi$ be a quantum channel. Then
			\begin{itemize}
				\item $\Phi$ is irreducible if it has a unique full-rank fixed state $\sigma$, i.e. $\Phi(\sigma)=\sigma$, with $0 < \sigma \in\mathcal{S}(\mathcal{H})$;
				
				\item $\Phi$ is called primitive if it is irreducible and it has trivial peripheral spectrum, namely $\spec_P(\Phi)=\{ 1\}$.
			\end{itemize} 
			\label{prim_def}
		\end{definition}
		\begin{proposition}
			Let $\Phi$ be a quantum channel. If $\Phi$ is primitive, then its asymptotic dynamics has no permutations, $\pi=\mathrm{id}$.
			\label{prim_suff}
		\end{proposition}
		\begin{remark}
			More generally, $\pi=\mathrm{id}$ also for quantum channels with one-dimensional attractor subspace.
			\label{no_per_rem}
		\end{remark}
		The proof follows immediately from the definition of primitive channels. Note that the properties of primitivity and infinite divisibility are independent, although both conditions guarantee $\pi=\mathrm{id}$. More precisely,
		\begin{equation}
			\Phi \mbox{ infinitely divisible} \quad \not\Rightarrow \quad \Phi \mbox{ primitive},
			\label{prim_not_nec}
		\end{equation}
		as confirmed by the one-state contraction channel~\eqref{contr_ch} with $\rho$ non-invertible. Conversely,
		\begin{equation}
			\Phi \mbox{ primitive} \quad\not\Rightarrow\quad \Phi \mbox{ infinitely divisible},
			\label{inf_not_nec}
		\end{equation}
		as it may be proved by taking into account the indivisible unital qubit channel~\cite{wolf2008dividing}
		\begin{equation}
			\Phi(X)=\frac{1}{3}(X^T + \tr(X)\mathbb{I}), \quad X\in \mathcal{B}(\mathbb{C}^2).
			\label{ind_ch_wolf}
		\end{equation}  
		It has the non-degenerate eigenvalue $\lambda=1$ and no other peripheral eigenvalues, namely it is primitive. Incidentally, the qubit channel $\Phi$ has minimum value of the determinant, $\det(\Phi)=-1/27$.
		Thus, from \eqref{prim_not_nec} and \eqref{inf_not_nec}, we can conclude that infinite divisibility and primitivity are not necessary conditions for $\pi=\mathrm{id}$. 
		
		Now, after characterizing and interpreting the absence of permutations in Eq.~\eqref{struc_phiP}, we conclude this Section with a last Proposition which 
		characterizes the unitarity and the Hilbert-Schmidt unitarity of an asymptotic map $\hat{\Phi}_P$. 
		\begin{proposition}
			Let $\Phi$ be a quantum channel with asymptotic map $\hat{\Phi}_P$ of the form~\eqref{struc_phiP}. Then,
			\begin{enumerate}
				\label{ch_unit_HS_unit}
				\item $\hat{\Phi}_P^\dagger=\hat{\Phi}_P^{-1}$, namely $\hat{\Phi}_P$ is Hilbert-Schmidt unitary iff $\| \rho_k \|_{\mathrm{HS}}=\| \rho_{\pi(k)}\|_{\mathrm{HS}}$, for all $k=1,\dots , M$.
				\item $\hat{\Phi}_P(X)=UX U^\dagger$ iff $\rho_k=V_k\rho_{\pi(k)}V_k^\dagger,\,k=1,\dots M$ for some unitaries $V_k: \mathcal{H}_{\pi(k),2} \rightarrow \mathcal{H}_{k,2}$.
			\end{enumerate}
		\end{proposition}
		The proof of the first characterization readily follows from the definition of Hilbert-Schmidt unitarity, whereas the proof of the second one was discussed in~\cite{AFK_asympt}.
		\begin{remark}
			From this result it is clear that $\pi=\mathrm{id}$ is not necessary for $\hat{\Phi}_P$ to be either a unitary channel or Hilbert-Schmidt unitary, as already discussed previously. See the example of the qubit channel~\eqref{unit_count}.
		\end{remark}
		\begin{remark}
			As it can be seen by using Theorem~\ref{th:extensiontheorem}
			\begin{equation}
				\| \rho_k \|_{\mathrm{HS}}=\| \rho_{\pi(k)}\|_{\mathrm{HS}} \quad\not\Rightarrow\quad \rho_k=V_k\rho_{\pi(k)}V_k^\dagger,
			\end{equation}
			so that
			\begin{equation}
				\hat{\Phi}_P^\dagger=\hat{\Phi}_P^{-1} \quad\not \Rightarrow\quad \hat{\Phi}_P(X)=UX U^\dagger.
				\label{unit_vs_Hs_unit}
			\end{equation}
		\end{remark}
		
	\end{section}

	\begin{section}{Asymptotics of faithful and non-faithful channels}
		\label{Jex_Wolf}
		
		The aim of this Section is twofold. First, we will generalize several results of~\cite{jex_st_2012,jex_st_2018} to the non-faithful case. Then, we will connect them to the structure theorem~\ref{Wolf_struc}. 
		
		Let us start with the following theorem, expressing mathematically the irreversibility of a non-unitary open-system dynamics~\cite{wolf2012quantum}.
		\begin{theorem}
			Let $\Phi$ be a quantum channel on $\mathcal{B}(\mathcal{H})$. Then $\Phi^\dagger=\Phi^{-1}$ iff $\Phi(X)=UXU^\dagger$, with $U$ unitary, i.e.\ $\Phi$ is a unitary channel.
		\end{theorem}
		\begin{remark}
			As already discussed in the previous section [see Eq.~\eqref{unit_vs_Hs_unit}], this is not true in general for the asymptotic map $\hat{\Phi}_P$.
		\end{remark}
		More generally, given a quantum channel $\Phi$ and a subset $\mathcal{M} \subseteq \mathcal{S}(\mathcal{H})$, we can look for another quantum channel $\Phi^{\prime}$ such that
		\begin{equation}
			\Phi^{\prime}\Phi(X)=X, \quad X \in \mathcal{M}.
		\end{equation}
		If such channel exists, then $\Phi$ is said to be sufficient (or reversible) with respect to the subset $\mathcal{M}$, and $\Phi^{\prime}$ is called the recovery map of $\Phi$ for $\mathcal{M}$~\cite{junge2018universal}.  Interestingly, this concept finds applications in quantum error correcting codes~\cite{ng2010simple}.
		
		For instance, at least for faithful channels with full-rank fixed state $\sigma$, it turns out that~\cite{jex_st_2018} 
		\begin{equation}
			\Phi\Phi^{\ddagger}(X)=\Phi^{\ddagger}\Phi(X)=X,\quad X\in \Attr(\Phi),
			\label{unit-1/2}
		\end{equation}
		where $\Phi^{\ddagger}$ is the adjoint map of $\Phi$ with respect to the scalar product
		\begin{equation}
			\begin{split}
				\braket{A}{B}_{1/2}&=\braket*{A}{\sigma^{-1/2}B\sigma^{-1/2}}_{\mathrm {HS}}=\tr(A^{\dagger}\sigma^{-1/2}B\sigma^{-1/2})=\tr(\sigma^{-1/2}A^{\dagger}\sigma^{-1/2}B)\\&=\tr(\sigma^{-1}\sigma^{1/2}A^{\dagger}\sigma^{-1/2}B)=\tr(\sigma^{-1}A^{\ddagger}B),\quad A,B\in \mathcal{B}(\mathcal{H}),
			\end{split}
			\label{1/2-sc}
		\end{equation}
		i.e.\ a generalized Hilbert-Schmidt scalar product involving a modified adjoint operation on $\mathcal{B}(\mathcal{H})$ defined by
		\begin{equation}
			B^{\ddagger}=\sigma^{1/2}B^{\dagger}\sigma^{-1/2},\quad B\in \mathcal{B}(\mathcal{H}).
		\end{equation}
		The map $\Phi^{\ddagger}$ explicitly reads
		\begin{equation}
			\Phi^{\ddagger}(X)=\sum_{k=1}^{N}\sigma^{1/2}A_{k}^{\dagger}\sigma^{-1/2} X \sigma^{-1/2}A_{k}\sigma^{1/2},\quad X\in \mathcal{B}(\mathcal{H}),
			\label{phidd_expr}
		\end{equation}
		with $\{ A_{k} \}_{k=1}^{N}$ being a system of Kraus operators of $\Phi$. It is clear that $\Phi^\ddagger$ is a quantum channel. Note that $\Phi^\ddagger$ is a particular example of Petz recovery maps~\cite{petz1986sufficient,petz1988sufficiency}.
		
		Furthermore, it is worthwhile to observe that 
		\begin{equation}
			\Attr(\Phi)=\Attr(\Phi^{\ddagger}),
			\label{attr_phi_phid} 
		\end{equation}
		as a consequence of a one-to-one correspondence between eigenspaces corresponding to peripheral eigenvalues of $\Phi$ and $\Phi^{\dagger}$ stated in the following Proposition~\cite[Section 4.3]{jex_st_2018}.
		\begin{proposition}
			Let $\Phi$ be a faithful quantum channel with full-rank fixed state $\sigma$. Then, if $\lambda$ is a peripheral eigenvalue, i.e.\ $|\lambda|=1$, 
			\begin{equation}
				X_{\lambda} \in \Ker(\Phi-\lambda\mathsf{1}) \quad\Leftrightarrow\quad \sigma^{-1/2}X_{\lambda} \sigma^{-1/2} \in \Ker(\Phi^{\dagger}-\bar{\lambda}\mathsf{1}).
			\end{equation}
			\label{eigen_phi_phid}
		\end{proposition}
		
		The following simple computation
		\begin{equation}
			\Phi^{\ddagger}(X_{\lambda})=\sum_{k=1}^{N}\sigma^{1/2} A_{k}^{\dagger}\sigma^{-1/2}X_{\lambda}\sigma^{-1/2}A_{k}\sigma^{1/2}=\bar{\lambda}\sigma^{1/2}\sigma^{-1/2}X_{\lambda}\sigma^{-1/2}\sigma^{1/2}=\bar{\lambda}X_{\lambda},
			\label{Kr_Phidd}
		\end{equation}
		thanks to Proposition~\ref{eigen_phi_phid} and the Kraus representation of $\Phi^{\dagger}$,
		\begin{equation}
			\Phi^{\dagger}(X)=\sum_{k=1}^{N}A_{k}^{\dagger} X A_{k},\quad X\in\mathcal{B}(\mathcal{H}),
		\end{equation}
		yields~\eqref{attr_phi_phid}.

		Now, the result~\eqref{unit-1/2} may be extended to the non-faithful case. To this purpose, define the map $\tilde{\Phi}:\mathcal{B}(\mathcal{H}_0)\rightarrow \mathcal{B}(\mathcal{H}_0)$ as follows:
		
		\begin{equation}
			\Phi(X)=0\oplus \tilde{\Phi}(\tilde{X}),\quad X=0\oplus \tilde{X}\in 0\oplus \mathcal{B}(\mathcal{H}_0).
			\label{dec_phi_phit}
		\end{equation}
		
		Note that $\tilde{\Phi}$ is a map between $\mathcal B(\mathcal H_0)$ and $\mathcal B(\mathcal H_0)$ because $0\oplus\mathcal B(\mathcal H_0)$ is invariant under $\Phi$. 
	
	It turns out that $\tilde{\Phi}$ is a faithful quantum channel, which we can call the induced faithful channel of $\Phi$. That being said, from~\eqref{Attr_struc} and~\eqref{dec_phi_phit} we have
		\begin{equation}
			\Attr(\Phi)=0\oplus \Attr(\tilde{\Phi}).
			\label{attr:phi_phit}
		\end{equation}
		Now, we can extend to the non-faithful case the definition of $\Phi^{\ddagger}$ by inserting a maximum-rank fixed state $\sigma$ with
		\begin{equation}
			\sigma^{-1}|_{\Ker(\sigma)}=\sigma^{-1}|_{\mathcal H_0^\perp}=0
		\end{equation}
		in~\eqref{phidd_expr}. The map $\Phi^{\ddagger}$ is a quantum operation, i.e. a trace non-increasing completely positive map, and, again, we can define the induced faithful channel $\widetilde{\Phi^{\ddagger}}$ of $\Phi^{\ddagger}$ via
		\begin{equation}
			\Phi^{\ddagger}(X)=0\oplus \widetilde{\Phi^{\ddagger}}(\tilde{X})=0\oplus \tilde{\Phi}^{\ddagger}(\tilde{X}),\quad X=0\oplus \tilde{X} \in 0 \oplus \mathcal{B}(\mathcal{H}_{0}),
		\end{equation}
		where the last step readily follows from the definitions. 
		Ultimately, the attractor subspace reads
		\begin{equation}
			\Attr(\Phi^{\ddagger})=0\oplus\Attr(\widetilde{\Phi^{\ddagger}}) =0\oplus\Attr(\tilde{\Phi}^{\ddagger}) =0\oplus \Attr(\tilde{\Phi})
			=\Attr(\Phi),
			\label{attr_rel}
		\end{equation}
	
		and~\eqref{unit-1/2} can be generalized to arbitrary channels.
		
		Now, let us turn our attention to the relation between the result~\eqref{unit-1/2} and the structure of the asymptotic map, Theorem~\ref{Wolf_struc}. To this purpose, let us define the double adjoint of the asymptotic map $\hat{\Phi}_{P}$, $\hat{\Phi}_{P}^{\ddagger}:\Attr(\Phi)\rightarrow \Attr(\Phi)$ by
		\begin{equation}
			\braket*{A}{\hat{\Phi}_{P}(B)}_{1/2}=\braket*{\hat{\Phi}_{P}^{\ddagger}(A)}{B}_{1/2},\quad A , B \in \Attr(\Phi),
		\end{equation}
		from which it is possible to obtain, using~\eqref{struc_phiP},
		\begin{equation}
			\hat{\Phi}_{P}^{\ddagger}(X)=0\oplus \bigoplus_{k=1}^{M} U_{\pi^{-1}(k)}^{\dagger}x_{\pi^{-1}(k)}U_{\pi^{-1}(k)} \otimes \rho_{k},\quad X \in \Attr(\Phi).
			\label{phipdd_expr}
		\end{equation}
		
		Then, it is immediate to check that
		\begin{equation}
			\hat{\Phi}_{P}^{\ddagger}\hat{\Phi}_{P}=\hat{\Phi}_{P}\hat{\Phi}_{P}^{\ddagger}=\mathsf{1}_{P}.
			\label{UNit-phip}
		\end{equation}
		So we can conclude that the asymptotic map $\hat{\Phi}_{P}$ is a unitary operator on the attractor subspace $\Attr(\Phi)$ with respect to the modified scalar product~\eqref{1/2-sc}, depending on the maximum-rank invariant state $\sigma$ of $\Phi$. 
		
		This is in line with the fact that $\hat{\Phi}_P$ is not generally Hilbert-Schmidt unitary and, consequently, a unitary channel, as explained in the previous Section. In addition, by comparing the expressions~\eqref{phiPd_expr} and~\eqref{phipdd_expr} for $\hat{\Phi}_P^\dagger$ and $\hat{\Phi}_P^\ddagger$ respectively, we obtain the first implication of Proposition~\ref{ch_unit_HS_unit}.  
		
	Also, observe that
		\begin{equation}
			(\widehat{\Phi^{\ddagger}})_{P}=\hat{\Phi}_{P}^{\ddagger},
			\label{P_ddagger}
		\end{equation}
		or, equivalently,
		\begin{equation}
			\Phi^{\ddagger}\Attr(\Phi)=\Attr(\Phi),
		\end{equation}
		as a consequence of~\eqref{attr_rel} and~\eqref{Attr_inva}.
		
		Now, let us prove that~\eqref{UNit-phip} is equivalent to~\eqref{unit-1/2}. Suppose that~\eqref{unit-1/2} holds. Then
		\begin{equation}
			\begin{split}
				\braket*{\Phi^{\ddagger}\Phi(X)}{Y}_{1/2}
				=\braket*{\hat{\Phi}_{P}(X)}{\hat{\Phi}_{P}(Y)}_{1/2}=\braket{X}{Y}_{1/2} ,\quad X,Y \in \Attr(\Phi),
			\end{split}
		\end{equation}
		which gives~\eqref{UNit-phip}, if one takes into account the invertibility of $\hat{\Phi}_{P}$. The converse implication follows immediately from~\eqref{P_ddagger}
		\begin{equation}
			\Phi^\ddagger\Phi(X)=(\widehat{\Phi^\ddagger})_P\hat{\Phi}_P(X)=\hat{\Phi}_P^\ddagger\hat{\Phi}_P(X)=X,\quad X\in \Attr(\Phi),
		\end{equation}
		and analogously for the other equality.
		To sum up, as anticipated, Theorem~\ref{Wolf_struc} implies~\eqref{unit-1/2}. 
		
	Moreover, from Theorems~\ref{Wolf_struc} and~\ref{eigen_phi_phid}, it follows that, in the faithful case,
		\begin{equation}
			\Attr(\Phi^{\dagger})=\sigma^{-1/2}\Attr(\Phi)\sigma^{-1/2}=\bigoplus_{k=1}^{M} \mathcal{B}(\mathcal{H}_{k,1}) \otimes \mathbb{I}_{k,2},
			\label{fin_alg}
		\end{equation}
		as observed in Remark~\ref{faith}.
		
	Incidentally, notice that in the non-faithful case $\Attr(\Phi^{\dagger})$ is not generally an algebra. More precisely, as an immediate consequence of~\cite[Proposition 3]{albert2019asymptotics}, given $\tilde{X}\in \Attr(\tilde{\Phi}^{\dagger})$ there exists $X\in \Attr(\Phi^{\dagger})$ such that $QXQ=\tilde{X}$, where $Q$ is the projection onto $\mathcal{H}_0$.
		
		In addition, Proposition~\ref{eigen_phi_phid} may be easily proved from Theorem~\ref{Wolf_struc}. The adjoint map $\Phi^{\dagger}$ of a faithful channel $\Phi$
	acts on its attractor subspace $\Attr(\Phi^{\dagger})$ as follows
		\begin{equation}
			\Phi^{\dagger}(X)=\bigoplus_{k=1}^{M} U_{\pi^{-1}(k)}^{\dagger}x_{\pi^{-1}(k)}  U_{\pi^{-1}(k)} \otimes \mathbb{I}_{k,2},\quad X=\bigoplus_{k=1}^{M} x_{k} \otimes \mathbb{I}_{k,2}\in \Attr(\Phi^{\dagger}).
		\end{equation}
		Now, the eigenvalue equation for $\Phi$
		\begin{equation}
			\Phi(X)=\lambda X,\quad |\lambda|=1,
		\end{equation}
		namely the condition
		\begin{equation}
			U_{k}x_{\pi(k)}U_{k}^{\dagger}=\lambda x_{k},\quad k=1,\dots , M,
		\end{equation}
		is equivalent to require that
		\begin{equation}
			U_{\pi^{-1}(k)}^{\dagger} \sigma_{\pi^{-1}(k)}^{-1/2} x_{\pi^{-1}(k)} \sigma_{\pi^{-1}(k)}^{-1/2} U_{\pi^{-1}(k)}= \bar{\lambda} \sigma_{k}^{-1/2} x_{k}\sigma_{k}^{-1/2},\quad k=1,\dots , M,
			\label{eigen_phid}
		\end{equation}
		thanks to~\eqref{max_rank_dec}.
		
		In conclusion,~\eqref{eigen_phid} is an explicit way of writing
		\begin{equation}
			\Phi^{\dagger}(\sigma^{-1/2}X \sigma^{-1/2})=\bar{\lambda} \sigma^{-1/2}X\sigma^{-1/2},
		\end{equation} 
		i.e.\ the statement. 
		Finally, note that similar bijective mappings between the eigenspaces $\Ker( \Phi-\lambda\mathsf{1} )$ and $\Ker(\Phi^{\dagger}-\bar{\lambda}\mathsf{1})$, established in Theorem 3.1 of~\cite{jex_st_2012} and Theorem 2 of~\cite{jex_st_2018} (see also Lemma 3 of~\cite{burgarth_ergodic}), may be proved in the same way by means of Theorem~\ref{Wolf_struc}.  
		
	\end{section}
	\begin{section}{Conclusions}
		\label{concl}	
		In this Article we have explored several aspects of the asymptotic dynamics of quantum channels of finite-dimensional systems. The starting point of our findings is a structure theorem that provides an expression for the asymptotic map. It is the restriction of the quantum channel generating the discrete-time dynamics to the attractor subspace, where the evolution takes place at large times.  
		
		Four main goals were achieved in this paper.
		First, a structure theorem of the asymptotic map in terms of disjoint cyclic permutations was given (Sec.~\ref{aternativeexpression}). Second, a quantum channel unfolding a given asymptotic map  
		was constructed, thus 
		finding the converse to the structure theorem (Sec.~\ref{extens_th}). Third, the role of permutations in the asymptotic dynamics was understood and characterized (Sec.~\ref{per_role}). Fourth, several properties of the asymptotics of quantum channels were derived from the structure theorem of the asymptotic map (Sec.~\ref{Jex_Wolf}). 
		
		For what concerns the third aim, the asymptotic dynamics is not purely unitary, since partial permutations between the factors of the attractor manifold may occur. However, the absence of permutations does not characterize
		the unitarity of the asymptotic map, but a weaker divisibility property, as stated in Proposition~\ref{ch_pi_id}. An almost immediate consequence of this result is that the asymptotic dynamics of an infinitely divisible channel, such as a Markovian or an idempotent channel, has no permutations. 
		
		Therefore, given a discrete-time semigroup $(\Phi^n)_{n\in\mathbb{N}}$ associated with a quantum channel $\Phi$, the asymptotic map will be unitary if we look at the coarse-grained dynamics generated by the quantum channel $\Phi^M$, with $M$ being the least common multiple of the lengths of the cycles of the permutation~\cite{AFK_asympt}. To conclude, the occurrence of permutations is related to long-term non-Markovian effects, making the asymptotic dynamics generally non-unitary. 
		Notice also that permutations can take place only in the presence of a degenerate attractor subspace (see Remark~\ref{no_per_rem}).
		It is also worthwhile to observe that, as revealed by Eq.~\eqref{UNit-phip}, permutations do not even guarantee the Hilbert-Schmidt unitarity of the asymptotic map, physically implying that the purity of the asymptotic states is not generally conserved under the dynamics. 
		
		These results shed light on the crucial role of the asymptotic dynamics of a quantum channel in quantum technology applications. With this respect it would be interesting to better understand its relation with entanglement, a fundamental resource in quantum information processes, and in particular its role in the class of eventually entanglement breaking~\cite{hanson2020eventually} and entanglement saving~\cite{lami2016entanglement} channels. Such maps generalize entanglement breaking channels, deeply studied in the last twenty years~\cite{horodecki2003entanglement,rahaman2018eventually,christandl2019composed} because of their detrimental effects in quantum communication protocols~\cite{bauml2015limitations,christandl2017private}.
		
	\end{section}
	
	\section*{Acknowledgments}
	This work was partially supported by Istituto Nazionale di Fisica Nucleare (INFN) through the project ``QUANTUM'' and the Italian National Group of Mathematical Physics (GNFM-INdAM). 

	\appendix
	
	\begin{section}{Technical proofs}
		\label{proofs}
		\begin{subsection}{Proof of Eq.~\eqref{suppPpsuppP}}
			\label{SUPPPSUPPP}
			We want to prove that
			\begin{equation}
				\label{eq:H=H'}
				\Ran \mathcal P_P(\mathbb I)=\supp \mathcal P_P(\mathbb I)=\mathcal H_0.
			\end{equation}
			We start by proving that
			\begin{equation}
				Q^\perp\mathcal{P}_{P}(\mathbb{I})Q^\perp=0,
				\label{th_1}
			\end{equation}
			where $Q^\perp$ is the projection onto the orthogonal complement $\mathcal{H}_{0}^\perp$ of $\mathcal{H}_0$. 
			Since $Q^\perp\mathcal{P}_{P}(\mathbb{I})Q^\perp \geqslant 0$, it is sufficient to prove that
			\begin{equation}
				\tr(Q^\perp\mathcal{P}_{P}(\mathbb{I})Q^\perp)=0.
				\label{th_2}
			\end{equation}
			We know that, by using~\eqref{P_for},
			\begin{equation}
				Q^\perp\mathcal{P}(\mathbb{I})Q^\perp=0 \quad\Leftrightarrow\quad \lim_{N\rightarrow \infty}\frac{1}{N} Q^\perp \sum_{n=1}^{N}\Phi^{n}(\mathbb{I})Q^\perp =0,
			\end{equation}
			which implies that
			\begin{equation}
				\lim_{N\rightarrow \infty}\frac{1}{N} \tr(Q^\perp \sum_{n=1}^{N}\Phi^{n}(\mathbb{I})Q^\perp ) =0.
			\end{equation}
			That being said, compute for a given $n\in\mathbb{N}$,
			\begin{equation}
				\begin{split}
					&\tr ( Q^\perp \Phi^{n+1}(\mathbb{I})Q^\perp )=\tr ( Q^\perp \Phi((Q+Q^\perp )\Phi^{n}(\mathbb{I})(Q+Q^\perp))Q^\perp)\\&=\tr (Q^\perp \Phi(Q\Phi^{n}(\mathbb{I})Q^\perp +Q^\perp\Phi^{n}(\mathbb{I})Q+Q^\perp\Phi^{n}(\mathbb{I})Q^\perp )Q^\perp)+\tr ( Q^\perp \Phi(Q\Phi^{n}(\mathbb{I})Q)Q^\perp) 
					\\&\leqslant \tr ( Q\Phi^{n}(\mathbb{I})Q^\perp +Q^\perp \Phi^{n}(\mathbb{I})Q )+\tr( Q^\perp\Phi^{n}(\mathbb{I})Q^\perp  ) = \tr ( Q^\perp\Phi^{n}(\mathbb{I})Q^\perp ),
				\end{split}
			\end{equation}
			where the inequality arises from the fact that the map $\Phi_{Q^\perp}(X)=Q^\perp X Q^\perp$ is trace non-increasing and $0\oplus\mathcal{B}(\mathcal{H}_0)$ is an invariant subspace for $\Phi$ (see Eq.~\eqref{dec_phi_phit}). So we can write
			\begin{equation}
				\frac{1}{n_{i}} \tr ( Q^\perp \sum_{n=1}^{n_i} \Phi^{n}(\mathbb{I})Q^\perp ) \geqslant  \frac{1}{n_{i}} n_i \tr ( Q^\perp\Phi^{n_i}(\mathbb{I})Q^\perp ) \geqslant 0,
			\end{equation}
			for the increasing subsequence $\{ n_i \}_{i\in\mathbb{N}}$ in~\eqref{P_p_for}, yielding~\eqref{th_2} by squeeze theorem and, consequently,~\eqref{th_1}. Finally, from equation~\eqref{th_1} we have, for $\phi \in \mathcal H$
			\begin{equation}
				\mel{\phi}{Q^\perp \mathcal P_P(\mathbb I)Q^\perp}{\phi}=\norm*{(\mathcal P_P(\mathbb I))^{1/2}Q^\perp\phi}^2=0,
			\end{equation}
			implying $(\mathcal P_P(\mathbb I))^{1/2}Q^\perp=0$. Then, by post-composing with $(\mathcal P_P(\mathbb I))^{1/2}$, we get $\mathcal P_P(\mathbb I)Q^\perp=0$, and, by taking the Hermitian conjugate, $Q^\perp\mathcal P_P(\mathbb I)=0$, i.e.~\eqref{eq:H=H'}.
		\end{subsection}
		\begin{subsection}{Proof of Eq.~\eqref{PpneqP}}
			\label{PPNEQP}
			By~\eqref{Jordan_dec} it is straightforward to prove that the equality $\mathcal{P}_P(\mathbb{I})=\mathcal{P}(\mathbb{I})$ is equivalent to
			\begin{equation}
				\Phi(\mathcal{P}_P(\mathbb{I}))=\mathcal{P}_P(\mathbb{I}),
				\label{eq:eigenvaluepp}
			\end{equation}
			which explicitly reads, in the faithful case,
			\begin{equation}
				\bigoplus_{k=1}^{M}  m_{\pi(k)}\mathbb{I}_{k,1} \otimes \rho_{k}= \bigoplus_{k=1}^{M}  m_{k}\mathbb{I}_{k,1} \otimes \rho_{k} \quad\Leftrightarrow\quad m_{\pi(k)}=m_k,\quad  k=1,\dots , M.
				\label{PpeqP}
			\end{equation}
			However, it is always possile to construct a channel for which~\eqref{PpeqP} does not hold through Theorem~\ref{th:extensiontheorem}.
		\end{subsection}
		\begin{subsection}{Complete positivity of the map $\Phi_0$ in~\eqref{eq:Phi0}}
			\label{CP_Phi_0}
			By definition of complete positivity, let us consider a $D$-dimensional Hilbert space $\mathcal H^\prime$, and check that the map
			\begin{align}
				\Phi_\mathrm{0}\otimes \mathsf{1}_{\mathcal B(\mathcal{H}^\prime)}:\mathcal{A}\otimes  \mathcal{B}(\mathcal{H}^\prime)\mapsto \mathcal{C}\otimes \mathcal{B}(\mathcal{H}^\prime)
			\end{align}
			is positive, with $\mathcal{A}$ and $\mathcal{C}$ denoting the algebras
			\begin{align}
				&\mathcal{A}=\mathcal B(\mathcal H_0^\perp)\oplus\bigoplus_{k=1}^M  \mathcal B(\mathcal H_{k,1}\otimes\mathcal H_{k,2}),\\
				&\mathcal{C}=0\oplus\bigoplus_{k=1}^M \mathcal B(\mathcal H_{k,1}\otimes\mathcal H_{k,2}).
			\end{align}
			Let us take a positive operator $Z$ over $\mathcal{A}\otimes \mathcal{B}(\mathcal{H}^\prime)$, and write it explicitly as
			\begin{equation}
				Z=\sum_{\alpha,\beta=1}^D Z_{\alpha\beta}\otimes \ketbra{\alpha}{\beta},
			\end{equation}
			with $Z_{\alpha\beta}\in\mathcal{A}$ and $\{\ket\alpha\}_{\alpha=1}^D$ an orthonormal basis of $\mathcal H^\prime$.  After expressing the vector state $\ket \phi \in \mathcal{H}\otimes \mathcal{H}^\prime$ as
			\begin{equation} 
				\ket \phi=\sum_{\alpha=1}^D \ket{\phi_\alpha}\otimes\ket{ \alpha}=\sum_{\alpha=1}^D\bigoplus_{k=0}^M\ket*{\phi_\alpha^{(k)}}\otimes\ket{ \alpha},
			\end{equation}
			for some $\ket {\phi_\alpha} =\bigoplus_{k=0}^M\ket*{\phi_\alpha^{(k)}}\in \mathcal{H} $,
			the positivity of $Z$ reads
			\begin{equation}
				\mel{\phi}{Z}{\phi}=\sum_{\alpha,\beta=1}^D\mel{\phi_\alpha}{Z_{\alpha\beta}}{\phi_\beta}=\sum_{\alpha,\beta=1}^D\sum_{k=0}^M\mel*{\phi_\alpha^{(k)}}{Z_{\alpha\beta}^{(k)}}{\phi_\beta^{(k)}}\geqslant 0.
				\label{Z_pos}
			\end{equation}
			Analogously
			\begin{align}
				\begin{split}
					\mel{\phi}{(\Phi_0 \otimes \mathsf{1}_{\mathcal{B}(\mathcal{H}^{\prime})})(Z)}{\phi}&=\sum_{\alpha,\beta=1}^D\mel*{\phi_\alpha}{\Phi_0(Z_{\alpha\beta})}{\phi_\beta } \\
					&=\sum_{\alpha,\beta=1}^D\sum_{k=1}^{M}\mel*{\phi_\alpha^{(k)}}{U_k\tr_{\pi(k),2} (Z_{\alpha\beta}^{(\pi(k))})U_k^\dagger \otimes \rho_k}{\phi_\beta^{(k)}}.
				\end{split}
				\label{CP_comp}
			\end{align}
			Now, let  
			\begin{align}
				\rho_k=\sum_{\ell=1}^{m_k}\lambda_{\ell}^{(k)}\ketbra*{\ell^{(k)}},\\
				\ket*{\phi_{\alpha}^{(k)}}=\sum_{\ell=1}^{m_k}\ket*{\phi_{\alpha \ell}^{(k)}} \otimes \ket*{\ell^{(k)}}
			\end{align}
			be the spectral decomposition of $\rho_k$ and the tensor product decomposition of $\ket*{\phi_{\alpha}^{(k)}}$ in terms of the basis $\ket*{\ell^{(k)}}$ of eigenvectors of $\rho_k$.  
			Therefore, the generic term in the last sum in~\eqref{CP_comp} becomes
			\begin{equation}
				\begin{split}
					&\mel*{\phi_{\alpha}^{(k)}}{U_k\tr_{\pi(k),2} (Z_{\alpha\beta}^{(\pi(k))})U_k^\dagger \otimes \rho_k}{\phi_{\beta}^{(k)}}\\
					&=\sum_{\ell,\ell'=1}^{m_k}\mel*{\phi_{\alpha \ell}^{(k)} \otimes \ell^{(k)}}{U_k\tr_{\pi(k),2} (Z_{\alpha\beta}^{(\pi(k))})U_k^\dagger \otimes \rho_k}{\phi_{\beta \ell'}^{(k)} \otimes \ell'^{(k)}}\\
					&=\sum_{\ell=1}^{m_k}\lambda_{\ell}^{(k)}\mel*{\phi_{\alpha \ell}^{(k)} }{U_k\tr_{\pi(k),2} (Z_{\alpha\beta}^{(\pi(k))})U_k^\dagger }{\phi_{\beta \ell}^{(k)}}\\
					&=\sum_{\ell,\ell'=1}^{m_k}\mel*{\psi_{\alpha \ell}^{(k)} \otimes \ell'^{(k)}}{ Z_{\alpha\beta}^{(\pi(k))}}{\psi_{\beta \ell}^{(k)} \otimes \ell'^{(k)}},
				\end{split}
			\end{equation}
			with
			\begin{equation}
				\ket*{\psi_{\alpha \ell}^{(k)}}=\sqrt{\lambda_\ell^{(k)}}U_k^\dagger\ket*{\phi_{\alpha \ell}^{(k)}}.
			\end{equation}
			Therefore,
			\begin{equation}
				\mel{\phi}{(\Phi_0 \otimes \mathsf{1}_{\mathcal{B}(\mathcal{H}^{\prime})})(Z)}{\phi}=\sum_{\ell,\ell'=1}^{m_k}\sum_{k=1}^M\sum_{\alpha,\beta=1}^D\mel*{\psi_{\alpha \ell}^{(k)} \otimes \ell'^{(k)}}{ Z_{\alpha\beta}^{(\pi(k))}}{\psi_{\beta \ell}^{(k)} \otimes \ell'^{(k)}}\geqslant 0,
			\end{equation}
			as a consequence of~\eqref{Z_pos}.
		\end{subsection}
		\begin{subsection}{Proof of Eqs.~\eqref{Attr_nth_root}-\eqref{eigen_nth_root}}
			\label{NTH-ROOT}
			Let us write the Jordan decompositions of $\Phi$ and $\Phi_{n}$
			\begin{align}
				&\Phi=\sum_{k=1}^N (\lambda_k \mathcal{P}_{k}+\mathcal{N}_{k}),\\
				&\Phi_n=\sum_{k=1}^{N^\prime}(\lambda_{k,n} \mathcal{P}_{k,n}+\mathcal{N}_{k,n}),
			\end{align}
			where $\mathcal{P}_{k}$ and $\mathcal{N}_{k}$ ($\mathcal{P}_{k,n}$ and $\mathcal{N}_{k,n}$) are the eigenprojection and eigennilpotent of $\Phi$ ($\Phi_{n}$) corresponding to the $k$-th eigenvalue $\lambda_{k}$ ($\lambda_{k,n}$). Since $\Phi=\Phi_n^n$, we obtain
			\begin{equation}
				\sum_{k=1}^N (\lambda_k \mathcal{P}_{k}+\mathcal{N}_{k})=\sum_{k=1}^{N^{\prime}} \left( \lambda_{k,n}^n \mathcal{P}_{k,n}+\sum_{m=0}^{n-1} \begin{pmatrix}n \\ m\end{pmatrix} \lambda_{k,n}^{m} \mathcal{N}_{k,n}^{n-m} \right).
				\label{dec_phi_phin}
			\end{equation} 
			It is straightforward to prove that the right-hand side of~\eqref{dec_phi_phin} is an alternative Jordan decomposition of $\Phi$, so from the uniqueness of the Jordan decomposition, we have
			\begin{align}
				&\lambda_k=\lambda_{k,n}^{n},\\
				&\mathcal{P}_{P}=\mathcal{P}_{P,n},
			\end{align}
			which are Eqs.~\eqref{Attr_nth_root} and~\eqref{eigen_nth_root} of the paper.
		\end{subsection}
		
	\end{section}
	
	\section*{References}

	\bibliography{mybibfile.bib}

	\bibliographystyle{unsrt}
	\bibliographystyle{is-bst}
	
\end{document}